\documentclass[letterpaper, 10 pt, conference]{ieeeconf} 

\IEEEoverridecommandlockouts 

\usepackage{amsmath,amssymb,amsfonts}

\usepackage{amsthm}

\newtheorem{theorem}{Theorem}

\newtheorem{remark}{Remark}
\newtheorem{assumption}{Assumption}
\newtheorem{definition}{Definition}

\newtheorem*{conjecture*}{Conjecture}
\newtheorem{prop}{Proposition}

\newtheorem{problem}{Problem}
\usepackage[linesnumbered]{algorithm2e}
\usepackage{graphicx}
\usepackage{textcomp}
\usepackage{xcolor}
\usepackage{breqn}
\usepackage{tikz}

\usepackage{cite}

\title{
\LARGE \bf Optimization with Zeroth-Order Oracles in Formation
}

\author{Elad Michael, Daniel Zelazo, Tony A. Wood, Chris Manzie, and Iman Shames
\thanks{E.~Michael, T.~A.~Wood, C.~Manzie, and I.~Shames are with the Department of Electrical and Electronic Engineering, University of Melbourne. 
{\tt\small \{eladm@student, wood.t@,manziec@,ishames@\}unimelb.edu.au}
}
\thanks{D.~Zelazo is with the Faculty of Aerospace Engineering, Israel Institute of Technology, Israel. {\tt\small dzelazo@technion.ac.il}}
}

\begin{document}

\maketitle
\thispagestyle{empty}
\pagestyle{empty}

\begin{abstract}
In this paper, we consider the optimisation of time varying functions by a network of agents with no gradient information. The proposed a novel method to estimate the gradient at each agent's position using only neighbour information. The gradient estimation is coupled with a formation controller, to minimise gradient estimation error and prevent agent collisions. Convergence results for the algorithm are provided for functions which satisfy the Polyak-\L{}ojasiewicz inequality. Simulations and numerical results are provided to support the theoretical results.
\end{abstract}


\section{INTRODUCTION}

In time varying optimisation tasks, the goal is to optimise a sequence of problems where each new objective is a variation of the previous. The time varying function can represent the position of a moving source, with measurements capturing signal strength. We assume that only measurements, with no higher order information, are available at each iteration. To compensate for the lack of gradient information, we consider a cooperating formation of agents, sharing information to minimise the time-varying function. In general, this may be posed as a sequence of independent optimization problems\cite{gutjahr2016lateral}. One could simply treat every new cost function as an entirely new optimization problem, although this may be computationally infeasible. Additionally, solving for the optimum at each iteration is unnecessary if it is sufficient to remain within some neighbourhood of the optimum at every iteration. If there is a limit on the variation of objective parameters between iterations, the solution of the previous iteration can be updated to approach of the solution of the current iteration. 

At each iteration, some amount of information about the changing function $f_k$ must be measured. Here we adopt the term \emph{p-th order oracle}\cite{daniel2019zeroth} to describe the type of available information. If $p=0$ then the zeroth-order oracle only makes available the current function value $f_k(x_k)$, and not any gradient or higher order information. Gradient descent makes use of a first-order oracle, Newton's method a second-order oracle, etc. We derive an iterative approach to track the optima of a changing cost function using zeroth-order oracles, and minimal assumptions on the behavior of $f_k$. This is similar to finite difference stochastic approximation (FDSA), except in this case the function $f_k$ can only be observed at the locations of the agents, rather than user chosen sample points. Therefore, the agents receive information from wherever their neighbours are to compute an approximate descent direction at each iteration. As the accuracy of the agent's gradient estimate is dependent on the geometry of its neighbours, we incorporate a formation control strategy to ensure the gradient estimation is accurate.\\ \indent In the area of online optimization, we will give a short review of the time varying optimization problems, but largely focus on the gradient free solutions which are more relevant to our formulation. Time varying optimization problems are well studied, frequently under the name Online Convex Optimization or OCO\cite{zinkevich2003online} \cite{hazan2016introduction}. A predictive/corrective method for OCO is presented in\cite{lesage2020predictive}, using gradient information and line search methods. Online convex optimisation with constraints is addressed by\cite{simonetto2017prediction}, with regret bounds and convergence results. These approaches use gradient information which we assume is unavailable in this formulation. The term bandit feedback is also used to describe this problem coupled with a zeroth-order oracle, as it conforms to a multi-armed bandit problem with convex costs\cite{daniel2019zeroth}. Regret bounds assuming compactness and convexity are derived in\cite{bubeck2012regret}, and a similar technique is used in \cite{agarwal2010optimal} with a multi-point estimate at each iteration for bandit feedback problems. A similar technique but using only a stochastic two point sampling each iteration is derived in\cite{shamir2017optimal}. These results utilise random or user chosen function sampling at each iteration, are entirely centralised, and assume convexity of the unknown cost functions. A network of zeroth-order oracles localizing the source of a static scalar field is examined in~\cite{khong2014multi}, with existence and convergence results, by assuming the existence of controllers with given properties.\\ \indent In this paper, we present a novel algorithm which combines the information from a network of zeroth oracles to optimise a time-varying cost function. We assume the agents follow single integrator dynamics, and construct a gradient estimate which uses only local information. As well, we provide a novel method of bounding the gradient estimation error, which has an interesting geometric interpretation. As such, we incorporate formation control, along with the gradient descent, to minimise the gradient estimation error. Both gradient estimation and formation control laws require only local information, leading to an entirely distributed approach. Additionally, we allow for a time-varying objective function, and the assumptions on the time-varying objective functions are only the Lipschitz continuity of the gradient and the Polyak-\L{}ojasiewicz inequality. These assumptions are weaker than many which are used to provide the linear convergence of gradient descent algorithms\cite{karimi2016linear}. \\\indent The paper is organised as follows. Section~\ref{sec:probForm} is devoted to basic assumptions on the time-varying function and agent dynamics. Section~\ref{sec:gradOracle} covers the approximation of the gradient given only zeroth-order information from an agent and its neighbours and derives an error bound on the gradient approximation. Section~\ref{sec:optAndForm} introduces and unifies formation control with the minimization. Finally, simulation and conclusions are covered in Section~\ref{sec:simulations}.

\section{Problem Formulation}\label{sec:probForm}

Consider a network of $n$ agents where $x_k^{i} \in \mathbb{R}^d$ denotes the position of the $i$-th agent for $i\in \{1,...,n\}$ at iteration $k$ in dimension $d$. Let $\mathcal{G}=(\mathcal{V},\mathcal{E})$ be the underlying graph of the network with the vertex set $\mathcal{V}=\{1,...,n\}$ and the edge set $\mathcal{E}\subseteq \mathcal{V}\times\mathcal{V}$. The edge set $\mathcal{E}$ captures the communication topology of the network, i.e. agent $i$ receives information from agent $j$ if $(i,j)\in\mathcal{E}$. Denote the neighbour set of each agent $i$ by $\mathcal{N}^{i}$ where $\mathcal{N}^{i}=\{j \mid (i,j) \in \mathcal{E}\}.$ 

In this paper, we consider the edges to be bidirectional, i.e., if $(i,j)\in\mathcal{E}$ then $(j,i)\in\mathcal{E}$. We begin with an assumption on the network structure.
\begin{assumption}\label{ass:networkProp}
Assume that the network is \emph{undirected}, \emph{connected}, and that $|\mathcal{N}^{i}| \geq d \;\forall\; i\in\mathcal{V}$. 
\end{assumption} 
If the network is disconnected, then each connected subnetwork would display the same behavior as is presented in this paper. The assumption that the neighbor set has cardinality greater than or equal to the dimension is necessary for the algorithm presented, and as the authors primarily envision physical applications ($2$ or $3$ dimensional), is not seen as a restrictive assumption.

The agents are modeled as single integrators, with dynamics
\begin{align}
x_{k+1}^{i} = x_k^{i} + u_k^{i}. \label{eq:dyn}
\end{align}
At each time instance $k$, each agent $i$ can measure $y_k^{i} = f_k(x_k^{i})$. Let $\mathcal{X}^*_k$ denote the set of minimisers of the time-varying function $f_k$. The following assumptions hold for the functions being minimised $f_k$.
\begin{assumption}(Differentiability and Lipschitz Gradient):\label{ass:Lipschitz}
The function $f_k:\mathbb{R}^d \rightarrow \mathbb{R}$ is continuously differentiable. The gradient is Lipschitz with constant $L$, that is there exists a positive scalar $L$ such that, $x\in\mathbb{R}^d,\; y\in\mathbb{R}^d$, $$|| \nabla f_k(x) - \nabla f_k(y) || \leq L||x-y||, $$ or equivalently
$$ f_k(y) \leq f_k(x) + \nabla f_k(x)^T(y-x) + \frac{L}{2}||y-x||^2.$$
\end{assumption}

We allow for the cost function $f_k$ to change at each iteration, however we make two assumptions on the changing cost functions.
\begin{assumption}(Polyak Condition):\label{ass:polyak}
There exists a positive scalar $s$ such that $$||\nabla f_k(x)||^2\geq 2s(f_k(x) - f_k(x^*_k))$$ where $x_k^*\in\mathcal{X}^*_k$ are the minimisers of $f_k$.
\end{assumption}
\begin{assumption}(Bounded Drift in Time):\label{ass:drift}
There exist positive scalars $\eta_0$ and $\eta^*$ such that $|f_{k+1}(x)-f_k(x)|\leq \eta_0$ for all $x\in\mathbb{R}^d$ and $|f_k(x^*_k)-f_{k+1}(x^*_{k+1})|\leq \eta^*$.
\end{assumption}
The problem of interest is given below.
\begin{problem}
Let $\mathcal{X}^*_k$ denote the set of minimisers of the time-varying function $f_k$. For a network of $n$ agents modeled in~\eqref{eq:dyn}, under Assumptions~\ref{ass:networkProp}-\ref{ass:drift}, find inputs $u^{i}_k$ for all agents $i\in\mathcal{V}$ given $\mathcal{Y}_k^{i}=\{y_k^{j} \mid j\in\mathcal{N}^i\cup\{i\}\}$, i.e., the set of measurements available to agent $i$ at iteration $k$, and a positive constant $M$ such that $||x_k^{i} - x_k^*||\leq M$ as $k\rightarrow \infty$ where $x^*_k = \arg\min_{x\in\mathcal{X}^*_k} \; \|x^i_k-x\|$.
\end{problem}

\section{Zeroth Order Network}\label{sec:gradOracle}

If, at each time step $k$, each agent $i$ was able to query an oracle and receive $\nabla f_k(x_k^{i})$, then a standard gradient descent method could be used to reach the set of minimisers. Motivated by this, we construct an approximate gradient oracle which combines the set of measurements from the agent and its neighbours to produce a descent direction at each iteration $k$.

Consider the directional derivative along the path from agent $i$ to agent $j\in\mathcal{N}^i$
\begin{align*}
\nabla_{ji} f_k(x_k^{i}) &= \frac{(x_k^{j}-x_k^{i})^T}{||x_k^{j} - x_k^{i}||} \nabla f_k(x_k^{i}).
\end{align*}
We construct an estimate of the gradient $\Lambda^{i}(x_k)$ with an error term $\epsilon^{ji}_k$ 
\begin{align}
\langle v_k^{ji},\Lambda^{i}(x_k)\rangle &= \frac{y_k^{j} - y_k^{i}}{||x_k^{ji}||}, \label{eq:approximation}\\
\langle v_k^{ji},\nabla f_k(x_k^{i})\rangle &= \langle v_k^{ji} , \Lambda^{i}(x_k)\rangle - \epsilon^{ji}_k. \label{eq:exact}
\end{align}
where we are using the shortening $x_k^{ji} = x_k^{j}-x_k^{i}$ to represent the difference vector and $v_k^{ji} = x_k^{ji}/||x_k^{ji}||$ as the unit vector in the difference's direction. We use $\langle u,v\rangle$ to denote the standard inner product when superscripts make $u^Tv$ cumbersome. Note that if the function $f$ was linear, \eqref{eq:approximation} would be the exact directional derivative with $\epsilon^{ji}_k=0.$ Using the estimate $\langle v_k^{ji}, \Lambda^{i}(x_k)\rangle$ and Assumption~\ref{ass:Lipschitz}, the error term $\epsilon^{ji}_k$ is bounded by
\begin{align}
y_k^{j} - y_k^{i} - \langle x_k^{ji}, \nabla f_k(x_k^{i}\rangle  &\leq \frac{L}{2}||x_k^{ji}||^2, \nonumber \\
\epsilon^{ji}_k &\leq \frac{L}{2}||x_k^{ji}||. \label{eq:deltaBound}
\end{align}

However, computing an approximation of $-\nabla f_k(x^i_k)$ to use as a descent direction with bound-able error requires more than just information in the $x^{ji}_k$ direction. We must use more than $1$ neighbour to construct the approximation $\Lambda^{i}(x_k)$ of the full gradient $\nabla f_k(x^i_k)$. At each time step, agent $i$ computes, as the approximate gradient,
\begin{align}
\Lambda^{i}(x_k) &= \left[\sum_{j\in\mathcal{N}^i}v_k^{ji} (v_k^{ji})^T\right]^{-1}\sum_{j\in\mathcal{N}^i}\frac{y_k^{j} - y_k^{i}}{||x_k^{ji}||}v_k^{ji}. \label{eq:gradientEstimate}
\end{align}
Note that if the sum of outer products on the left of~\eqref{eq:gradientEstimate} is not of appropriate rank, it cannot be inverted to estimate the gradient. Additionally, if any adjacent agents coincide, then the gradient estimate $\Lambda^{i}(x_k)$ cannot be computed. Both the rank requirement and the requirement that no agents coincide will be addressed using formation control strategies in Section~\ref{sec:optAndForm}.
\begin{remark}\label{rem:Lambda_i}
Each agent $i$ can compute a local estimate of $\nabla f_k(x^i_k)$ via \eqref{eq:gradientEstimate} using only $x^i_k$ and $x^j_k$, $j\in\mathcal{N}^i$.
\end{remark}

In order to ensure that it is always possible to construct a gradient estimate, the formation control covered in Section~\ref{sec:optAndForm} will preventing neighbours from being co-linear or co-planar.

\begin{theorem}\label{thm:boundedEstimationError}
For a function $f_k$ and a set of agents satisfying Assumptions~\ref{ass:networkProp}-\ref{ass:drift}, the vector $\Lambda^{i}(x_k)$ as defined in~\eqref{eq:gradientEstimate} satisfies
\begin{align}
||\Lambda^{i}(x_k) - \nabla f_k(x_k^{i})|| \leq \delta_k^{i},
\end{align}
where, in $\mathbb{R}^2$, the error $\delta_k^{i}$ is defined to be
\begin{align}
\delta_k^{i} &:= \min_{j,l\in\mathcal{N}^i}\frac{L}{|\langle v_k^{li}, \bar{v}_k^{ji}\rangle|} \max(||x_k^{ji}+x_k^{li}||,||x_k^{jl}||), \label{eq:parallelogramBound}
\end{align}
and we have used $\bar{v}^{ji}$ to indicate a vector which is orthogonal to $v^{ji}$.
\end{theorem}

\begin{proof}
Recall the error bound on a single directional derivative state in~\eqref{eq:deltaBound}. Let $a^{ji}_k := \frac{L}{2}||x_k^{ji}||$ and $d^{ji}_k := \frac{y_k^{j} - y_k^{i}}{||x_k^{ji}||}$. Rearranging the error bound into a set of inequalities yields
\begin{align}
  d^{ji}_k - a^{ji}_k \leq \langle v_k^{ji},\nabla f_k(x_k^{i})\rangle \leq d^{ji}_k + a^{ji}_k \label{eq:setFromOneNeighbour}.
\end{align}
In $\mathbb{R}^2$ representing the space of all possible gradients, these two inequalities enclose a band of $\mathbb{R}^2$ of width $2a^{ji}_k$ bordered by two parallel lines perpendicular to $v_k^{ji}$. Consider the set of inequalities from an additional neighbour $l\in\mathcal{N}^i$
\begin{align}
  d^{il} - a^{li}_k \leq \langle v_k^{li}, \nabla f_k(x_k^{i})\rangle \leq d^{il} + a^{li}_k \label{eq:setFromTwoNeighbours},
\end{align}
Then as long as $v_k^{li}$ and $v_k^{ji}$ are not parallel they enclose a finite area parallelogram $\mathcal{P}_{ijl}\in\mathbb{R}^2$, and $\nabla f_k(x_k^{i})\in\mathcal{P}_{ijl}$. Note that assuming $v_k^{li}$ and $v_k^{ji}$ are not parallel is equivalent to Assumption~\ref{ass:drift}. The gradient $\nabla f_k(x^i_k)$ is inside the parallelogram $\mathcal{P}_{ijl}$ because it satisfies~\eqref{eq:setFromOneNeighbour} and~\eqref{eq:setFromTwoNeighbours}. From the original definition of the gradient estimate $\Lambda^{i}(x_k)$ in~\eqref{eq:approximation} we have
\begin{align}
  ||\langle v_k^{ji} , \Lambda^{i}(x_k) \rangle - d^{ji}_k|| = 0 \leq a^{ji}_k,
\end{align}
so $\Lambda^{i}(x_k)$ is inside $\mathcal{P}_{ijl}$ well. Therefore, the error $||\Lambda^{i}(x_k) - \nabla f_k(x_k^{i})||$ is bounded by diameter of the smallest ball containing $\mathcal{P}_{ijl}$. The diagonals of $\mathcal{P}_{ijl}$ have lengths
\begin{align}
l = \frac{L}{|\langle v_k^{li}, \bar{v}_k^{ji}\rangle|}||(x_k^{j}-x_k^{i})\pm(x_k^l-x_k^{i})||. \label{eq:diagLengths}
\end{align}
To upper bound the error, the longer diagonal is used. For an agent $i$ with neighbours $j$ and $l$, the longest diagonal thus has length
\begin{align}
  l &= \frac{L}{|\langle v_k^{li} , \bar{v}_k^{ji}\rangle|} \max(||x_k^{ji}+x_k^{li}||,||x_k^{jl}||),
\end{align}
which is the bound used in the Theorem.
\end{proof}

An example of the parallelogram $\mathcal{P}_{ijl}$ is shown in Figure~\ref{fig:2neighbour}, generated with $v_k^{ji} = \frac{1}{\sqrt{2}}[1,1]^T,\;a_k^{ji} = 2,\;d_k^{ji} = 0$, and $v_k^{li} = \frac{1}{\sqrt{5}}[1,-2]^T,\;a_k^{li} = 2,\;d_k^{li} = 3$.

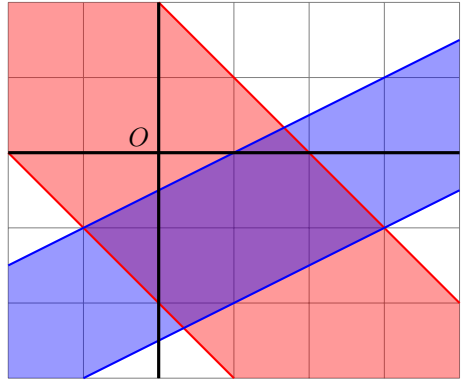
\begin{figure}
 \centering
    \begin{tikzpicture}
    \draw[step=1cm,gray,very thin] (-2,-3) grid (4,2);
  
    \fill [red, fill opacity=.4] 
        (-2,0) -- (-2,2) -- (0,2) -- (4,-2) -- (4,-3) -- (1,-3) -- cycle;
     \fill [blue, fill opacity=.4] 
        (-1,-3) -- (-2,-3) -- (-2,-1.5) -- (4,1.5) -- (4,-0.5)  -- cycle;
    
    \draw[thick, color=red] (-2,0) -- (1,-3);
    \draw[thick, color=red] (0,2) -- (4,-2);
    \draw[thick, color=blue] (-1,-3) -- (4,-.5);
    \draw[thick, color=blue] (-2,-1.5) -- (4,1.5);
    \draw[very thick] (-2,0) -- (4,0);
    \draw[very thick] (0,-3) -- (0,2);
   \draw (0 cm,1pt) -- (0 cm,-1pt) node[anchor=south east] {$O$};
\end{tikzpicture}
\caption{The set of feasible gradients which satisfy~\eqref{eq:setFromOneNeighbour} and~\eqref{eq:setFromTwoNeighbours}. \label{fig:2neighbour}}
\end{figure}

This bound does not take into account where within the parallelogram the gradient estimate $\Lambda^{i}(x_k)$ falls, which may decrease the distance to the farthest point by up to a factor of $2$. If the bound at each iteration $k$ is of interest, it is straightforward to check which corners of the parallelogram the gradient estimate is farthest from. If only the two neighbours $j,l\in\mathcal{N}^i$ which are used to calculate the bound in~\eqref{eq:parallelogramBound} are used to calculate the estimate~\eqref{eq:gradientEstimate}, then it is straightforward to see the estimate \emph{is} the center of the parallelogram and the bound is conservative by a factor of $2$.

An almost identical bounding procedure for the error is possible in $\mathbb{R}^n$, with each neighbour specifying a pair of parallel hyper plane constraints, which given $n$ neighbours result in an n-parallelotope. The approximate gradient formulation~\eqref{eq:gradientEstimate} is the same for any dimension.

Repeating this technique of partitioning the space of possible gradients with information from additional neighbours, the bound can be tightened. However, using additional agents significantly increases the computational burden, as the resulting polytope of possible gradients will have uncertain structure, and maximizing the norm under linear constraints is itself an NP-hard problem. Additionally, empirically the error bound from the complete set of neighbours largely seems to be determined by the pair of neighbours from Theorem~\ref{thm:boundedEstimationError}. Finally, computing the bound involving only $2$ neighbours is independent of the function measurements $y^i_k,y^j_k,y^l_k$. With additional neighbours forming a polytope with more facets, this computational convenience is lost.

Using the gradient approximation from~\eqref{eq:gradientEstimate}, the agents are able to find an approximate descent direction which has an error bounded by~\eqref{eq:parallelogramBound}. To minimise the bound~\eqref{eq:parallelogramBound}, thereby creating a better gradient approximation, we combine formation control with the decentralised minimization.

\section{Optimization in Formation}\label{sec:optAndForm}

Examining the parallelogram $\mathcal{P}_{ijl}$ in Fig.~\ref{fig:2neighbour}, there are two intuitive methods to minimise the diameter of the smallest bounding ball. We can bring the parallel edges closer together, ``thinning'' the parallelogram, and ensure that the two bands are orthogonal, ``squaring'' the parallelogram. These two methods correspond to maximizing the inner product in the denominator in~\eqref{eq:parallelogramBound}, ``squaring'' the parallelogram, and minimizing the distances term, ``thinning'' the parallelogram. The former can be achieved by keeping the vectors $v_k^{li}$ and $v_k^{ji}$ orthogonal, i.e. ensuring that $\langle v_k^{li} , \bar{v}_k^{ji}\rangle =1$ and the latter by keeping the agents as close as possible while maintaining a non-collision guarantee. Finally, it is critical to prevent violation of Assumption~\ref{ass:drift}, where $\langle v_k^{li} , \bar{v}_k^{ji}\rangle = 0$, which geometrically corresponds to both bands in Fig.~\ref{fig:2neighbour} being parallel to each other. We use decentralised navigation functions\cite{tanner2005formation,dimarogonas2010analysis} to maintain a desirable formation while minimizing $f_k$.
\begin{definition}\label{def:navFuncs}
Let $\phi^i:\mathbb{R}^{n_i d} \rightarrow \mathbb{R}^+$ be the navigation potential function for agent $i$ where $n_i = |\mathcal{N}_i| + 1$, with the following properties:
\begin{enumerate}
	\item The function $\phi^i$ is continuously differentiable on $\mathbb{R}^d$. 
	\item The function $\phi^i$ has a unique minimum, only attained when the agents are in the desired formation configuration.
	\item The function $\phi^i$ is Morse (critical points are non-degenerate).
	\item The function can be computed decentrally, i.e., each agent $i$ can compute $\phi^i(x_k)$ using only $x^i_k$ and $x^j_k$, $j\in\mathcal{N}^i$.
\end{enumerate}
\end{definition}
Note that these navigation potential functions exclude distance based approaches such as in~\cite{gazi2007aggregation}, as they are not Morse and we cannot, as of yet, prove the global convergence properties derived here. Using the decentralised navigation functions $\phi^i(x_k)$, and information available locally to each agent $i$, the agents are able to decrease a common global potential function
\begin{align}
\phi(x_k) = \sum_{i\in\mathcal{V}} \phi^{i}(x_k). \label{eq:globalPotential}
\end{align}
\begin{remark}\label{rem:phi_i}
To evaluate $\phi^{i}(x_k)$, agent $i$ needs access only to $x^i_k$ and $x^j_k$, $j\in\mathcal{N}^i$. No information about the position of all other agents is required. Consequently, $i$ can compute $\nabla_{i}\phi(x_k)$ 
using $x^i_k$ and $x^j_k$, $j\in\mathcal{N}^i$.\end{remark}
We make the following assumption throughout the remainder of the paper.
\begin{assumption}\label{ass:phiProperties}
The global potential function $\phi(x_k)$ is continuously differentiable, and the gradient is Lipschitz with constant $L_\phi$.
\end{assumption}
Defining the control input $u_k^{i}:= -\alpha^i_k p^{i}_k$, the dynamics are
\begin{align}
x^{i}_{k+1} = x^{i}_k - \alpha^i_k p^{i}_k,\label{eq:dynWgrad}
\end{align}
where $\alpha^i_k>0$ is a design constant. Define $p^{i}_k$ for agent $i$ at $k$ as
\begin{align}
p^{i}_k = \lambda^{i}_k \Lambda^{i}(x_k) + (1-\lambda^{i}_k) \nabla_{i}\phi(x_k), \label{eq:stepDirection}
\end{align}
where $\Lambda^{i}(x_k)$ is the estimate of the gradient from~\eqref{eq:gradientEstimate} and $\lambda^{i}_k\in[0,1]$ allows the agents to ``focus'' on the primary goal of minimizing $f_k$ while maintaining formation. The rules for deciding the weight $\lambda^{i}_k$ and constant $\alpha^i_k$ are laid out in Theorem~\ref{thm:formationThm}. 
\begin{theorem}\label{thm:formationThm}
Let $\phi(x) = \sum_{i\in\mathcal{V}}\phi^{i}(x_k)$ be the sum of functions $\phi_i(x_k)$, with a Lipschitz continuous gradient with constant $L_\phi$. Let  $\Phi^{i}$, $i\in\{1,\dots,n\}$ be positive constants. For a set of agents with dynamics as in~\eqref{eq:dynWgrad}, with step direction~\eqref{eq:stepDirection}, define the weighting parameter $\lambda^{i}_k$
\begin{align}
\lambda^{i}_k := \min(1, \frac{||\nabla_{i}\phi(x_k)||}{||\Lambda^{i}(x_k) - \nabla_{i}\phi(x_k)||} \frac{\sigma(\Phi^{i})}{\sigma(\phi^{i}(x_k))}),\label{eq:lambdaDetermination}
\end{align}
where $\sigma$ is a class $\mathcal{K}$ function. Define $\bar{\alpha}^{i}_k$ to be
\begin{align}
\bar{\alpha}^i_k = \frac{2c}{||\Lambda^{i}(x_k) - \nabla_{i}\phi(x_k)||^2},
\end{align}
where $c$ is a constant. Let the design constant $\alpha^i_k$ be in the interval
\begin{align}\label{eq:alpha_i}
\alpha^i_k\in(0,\min(\frac{1}{L_\phi},\frac{1}{L},\bar{\alpha}^i_k)]
\end{align}
where $L$ is the Lipschitz constant for the gradient of $f_k$. Then the system is stable, and $\exists k_0$ such that $\forall k\geq k_0$ the global potential function is bounded $\phi(x_k) \leq \sum_i \Phi^i + c$.
\end{theorem}

\begin{proof}
See Appendix~\ref{app:proofFormationThm}.
\end{proof}
\begin{remark}
In light of \eqref{eq:lambdaDetermination} and \eqref{eq:alpha_i}, and Remarks \ref{rem:Lambda_i} and \ref{rem:phi_i},  it can be seen that each agent $i$ can compute $u^i_k$ using only $x^i_k$ and $x^j_k$, $j\in\mathcal{N}^i$.
\end{remark}
Let $D^i_0$ be the projection of set $\{ x \mid \phi^i(x) \leq \Phi^i + c\}$ onto the subspace defined by $x^i$ and $x^j$, $j\in\mathcal{N}^i$. The boundedness of $\phi(x_k)$ corresponds to trajectories converging to $D^i_0$. We should choose the formation potential functions $\phi^i$ and design constants $\Phi^i,c$ such that having two collinear neighbours (in 2-D) or three coplanar neighbours (in 3-D) is impossible. Thus, we ensure that the matrix in~\eqref{eq:gradientEstimate} is full-rank. 

We may also leverage the convergence of the potential function to bound the gradient estimate error, as the formation fixes the geometry of the estimation. Using the new step direction, the the modified gradient error term is
\begin{align}
||\nabla f_k(x^{i}_k) - p^{i}_k|| &\leq \delta_k^i + (1-\lambda^{i}_k)||\Lambda^{i}(x_k) - \nabla_{i} \phi(x_k)||,\\
&\leq \delta_k^i + ||\Lambda^{i}(x_k)|| +  ||\nabla_{i} \phi(x_k)||,
\end{align}
where $\delta_k^i$ is the upper bound on the error of the original gradient estimate from~\eqref{eq:parallelogramBound}. Define $\rho^i$ be the radius of the smallest ball which contains $D^i_0$ and is centred at $x^i$. By the Lipshitz gradient property of $\phi(x)$, we then have 
\begin{align}
||\nabla f_k(x^{i}_k) - p^{i}_k|| &\leq \delta_k^i + ||\Lambda^{i}(x_k)|| +  L_{\phi}\rho^i \; \forall\;k\geq k_0
\end{align}
Finally if we assume that $||\Lambda^{i}(x_k)||\leq \gamma^i \; \forall \;k\geq k_0$,
\begin{align}
||\nabla f_k(x^{i}_k) - p^{i}_k|| &\leq \delta_k^i + \gamma^i +  L_{\phi}\rho^i \; \forall\;k\geq k_0 \label{eq:gradErrorWFormation}
\end{align}

To achieve the desired cooperative tasks agent $i$ executes the following steps at each $k$, (i) the gradient of $f_k$ at $x^{i}_k$ is estimated using \eqref{eq:gradientEstimate}; (ii) $\nabla_{i} \phi(x_k)$ is computed; (iii) the value of $\lambda^{i}_k$ is chosen via~\eqref{eq:lambdaDetermination}; (iv) the state is updated through~\eqref{eq:stepDirection} and an appropriate choice of $\alpha^i_k$.

We conclude this section by commenting on the overall performance of the agents in tracking the minimiser(s) of $f_k$. To this end, note that the directions $p_k$ at each iteration are still only approximations of the true gradients. The formation of zeroth-order agents cooperating is thus equivalent to individual agents querying a $\delta-$\emph{first order oracle} at each iteration $k$. The definition of a $\delta-$\emph{first order oracle} is given in Definition~\ref{def:deltaOracle}.
\begin{definition}\label{def:deltaOracle}
\emph{($\delta$-first order Oracle):} Given the function $f$ and a point $x$ the oracle returns $p(x) = \nabla f(x) + \delta(x)$ such that $||\delta(x)|| \leq \bar{\delta}$ for some positive scalar $\delta$.
\end{definition}

Here we show that a $\delta$-first order oracle is sufficient to converge to a neighbourhood of the minimisers $\mathcal{X}^*_k$, using~\eqref{eq:gradErrorWFormation} to construct an error bound on the $\delta$-first order Oracle
\begin{align}
\bar{\delta}^i := \delta^i_k + \gamma^i +  L_{\phi}\rho^i. \label{eq:deltaBarDef}
\end{align}
With the bounds introduced in~\eqref{eq:gradErrorWFormation}, we may also define a constant $\alpha$ for each agent,
\begin{align}
\alpha^i \in (0,\min(\frac{1}{L_\phi},\frac{1}{L_\phi},\frac{2c}{\gamma^i(\gamma^i + 2L_{\phi}\rho^i)})],
\end{align}
which satisfies all of the required properties for Theorem~\ref{thm:formationThm}. Note that the $\bar{\delta}$ used in Proposition~\ref{prop:neighbourhood} includes the formation control term $L_{\phi}\rho^i$, because it is an additional error in the gradient estimate, although it benefits the network as a whole.

\begin{prop}\label{prop:neighbourhood}
If $\alpha^i$ is chosen such that $|(1-\alpha^i s)| < 1$, then an agent using the $\delta$-first order oracle will reach an $M = \frac{\eta_0 + \eta^*}{2\alpha^i s^2} + \frac{(\bar{\delta}^i)^2}{4s^2}$ neighbourhood of the optimiser $\mathcal{X}_k^*$ as the time steps $k\rightarrow\infty$.
\end{prop}

\begin{proof}
See Appendix~\ref{app:proofProp}.
\end{proof}


\section{Simulations}\label{sec:simulations}


In this section we will implement, illustrate, and analyze the method described in the previous sections. We use a formation potential adapted from~\cite{tanner2005formation}, where each agent uses the following potential function
\begin{align}
\phi^{i}(x_k) = \frac{\sum_{j\in\mathcal{N}^i} ||x^i_k-x^j_k - c^{ij}||^2_2 }{e^{\beta(x_k)}}. \label{eq:formationPotential}
\end{align}
In the formation potential function given in~\eqref{eq:formationPotential}, fully explored in~\cite{tanner2005formation}, the numerator is a quadratic attraction potential to the desired difference between agents $i$ and $j$. The function in the denominator $\beta(x_k)$ is described as a ``collision function'', which is nominally equal to 1 but quickly vanishes as the agents reach a prescribed safety distance of each other or an obstacle. The decentralised formation control from~\cite{tanner2005formation,dimarogonas2010analysis} is shown to almost always converge, except from a set of initial conditions with measure zero. We have chosen the desired displacements $c^{ij}$ to form a hexagon with side lengths $s=4$. The gradient error bound~\eqref{eq:parallelogramBound} is 
\begin{align*}
||\Lambda^{i}(x_k) - \nabla f_k(x_k^{i})|| \leq \frac{L||x_k^{ji}+x_k^{li}||}{|\langle v_k^{li} , \bar{v}_k^{ji} \rangle|} = 2sL + \epsilon(\lambda_0,\Phi^{i}).
\end{align*}
If the agents were in a perfect hexagon formation, $2sL$ would be their error bound, for $s$ the side length and $L$ the Lipschitz constant. However, the formation maintenance is balanced with the minimization goal, so $\epsilon(\lambda_0,\Phi^{i})$ represents the additional error introduced from the choices of nominal weight $\lambda_0$ and the acceptable deviation bound $\Phi^{i}$. The specific choices of these parameters, and their impacts, are examined further in this section. Each agent's eventual nearest neighbours in the hexagon are its neighbour set $\mathcal{N}$.

For the function to be minimised $f_k$, we use convex quadratic function in two dimensions,
\begin{align*}
f_k(x_k^{(i)}) &= x_k^{(i)T} Q x_k^{(i)} + \zeta(t)^T x_k^{(i)},
\end{align*}
with $Q \succeq 0$. The randomly generated quadratic in the following examples is,
\begin{align*}
Q =
\begin{bmatrix}
3.89 & 0.45 \\
0.45 & 5.86
\end{bmatrix},
\end{align*}
to 2 decimal places. To simulate a moving source, the linear term $\zeta(t)$ is used to translate the quadratic along a path in the plane at a constant speed. The nominal weighting between the formation gradient and minimization gradient was $\lambda_0=1$, i.e. fully weighted on the minimization. This ensures that as long as the formation is ``good enough'', the agents will be attaining the best gradient estimate. The class $\mathcal{K}$ function $\sigma(\phi(x))$ used in~\eqref{eq:lambdaDetermination} is $\sigma(z) = z^2$, and the upper bounds for all agents potential functions' are $\Phi^{i}=1$. The trajectories of the agents and source function are shown in Figure~\ref{fig:trajectPhi1}, with the dots and star symbolizing the final position of the agents and optimum of $f_k$.
\begin{figure}[thpb]
 \centering
 \includegraphics[scale=0.55]{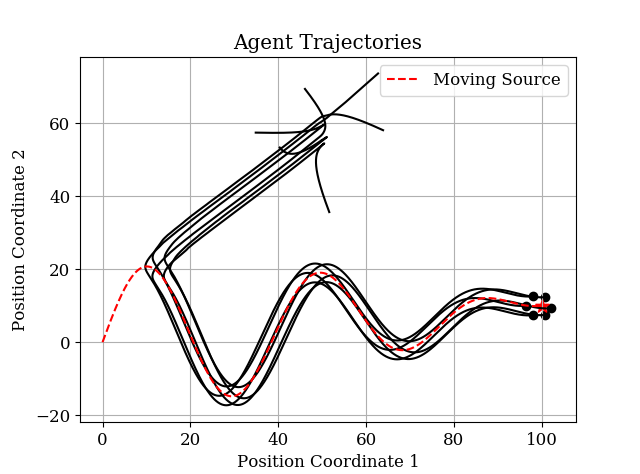}
 \caption{Agent trajectories with upper bound $\Phi=1$. \label{fig:trajectPhi1}}
\end{figure}

Immediately after the random initialization, the agents are not in formation. Therefore, their individual formation potential values $\phi^i(x_k)$ far exceed the prescribed upper bound $\Phi^i$, and they coalesce into formation. Once in formation, or ``close enough'' as determined by $\phi^i(x_k) \leq \Phi^i$, the formation begins converging to the neighbourhood of the minimisers of $f_k$. The minimization error $f_k(x_k) - f_k(x^*_k)$ and neighbourhood bound from~\eqref{eq:recursive} are shown for individual agents in Fig.~\ref{fig:minError}. 
\begin{figure}[thpb]
 \centering
 \includegraphics[scale=0.11]{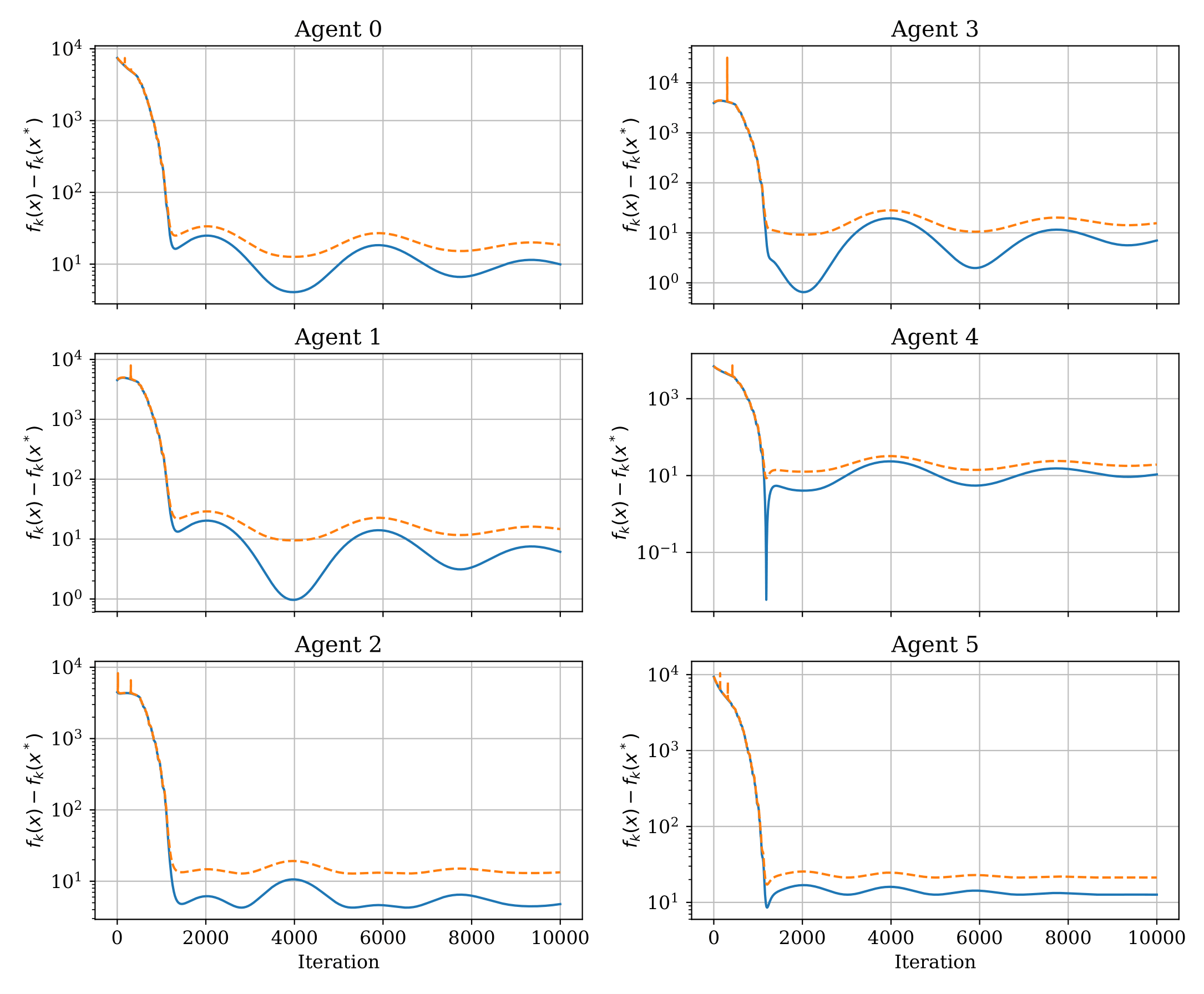}
 \caption{Minimization error with $\Phi=1$. \label{fig:minError}}
\end{figure}
The agents quickly converge to formation around the minimiser, and remains within the neighbourhood, oscillating beneath the bound as the source function $f_k$ changes. If the upper bound $\Phi^{i}$ was decreased, representing a more stringent requirement on the formation control, the formation would converge to the minimisers $\mathcal{X}^*_k$ more slowly. The choice of upper bounded is also clearly tied with the choice of the potential function $\phi(x)$. If there is a critical safety distance, between UAVs for example, then the upper bound $\Phi^i$ must be chosen such that the individual safety distance corresponds to a potential function value \emph{greater than} the upper bound $\Phi^{i}+c$. 

To demonstrate the benefits of the formation control, Fig.~\ref{fig:trajNoForm} shows the trajectories of the same agents without any formation control. The configuration of the agents is significantly looser, and though minimization error in Fig.~\ref{fig:minErrorNoForm} converges more quickly, it is approximately an order of magnitude larger than in Fig.~\ref{fig:minError}. The error of the gradient estimate is high in this case largely due to the distance between the agents being significantly more than in the formation controlled case Fig.~\ref{fig:trajectPhi1}.
\begin{figure}[thpb]
 \centering
 \includegraphics[scale=0.55]{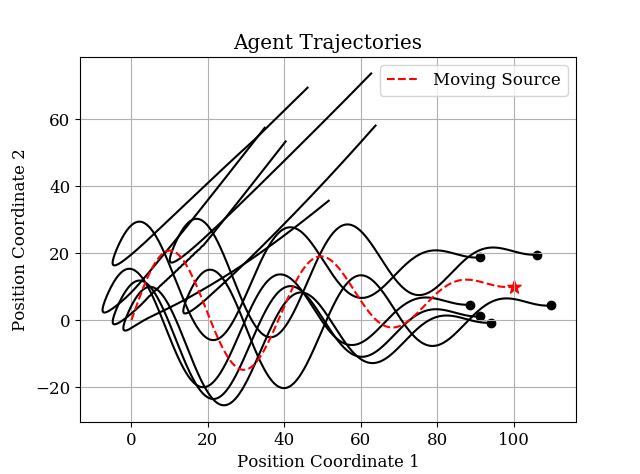}
 \caption{ Agent trajectories with no formation. \label{fig:trajNoForm}}
\end{figure}

\begin{figure}[thpb]
 \centering
 \includegraphics[scale=0.11]{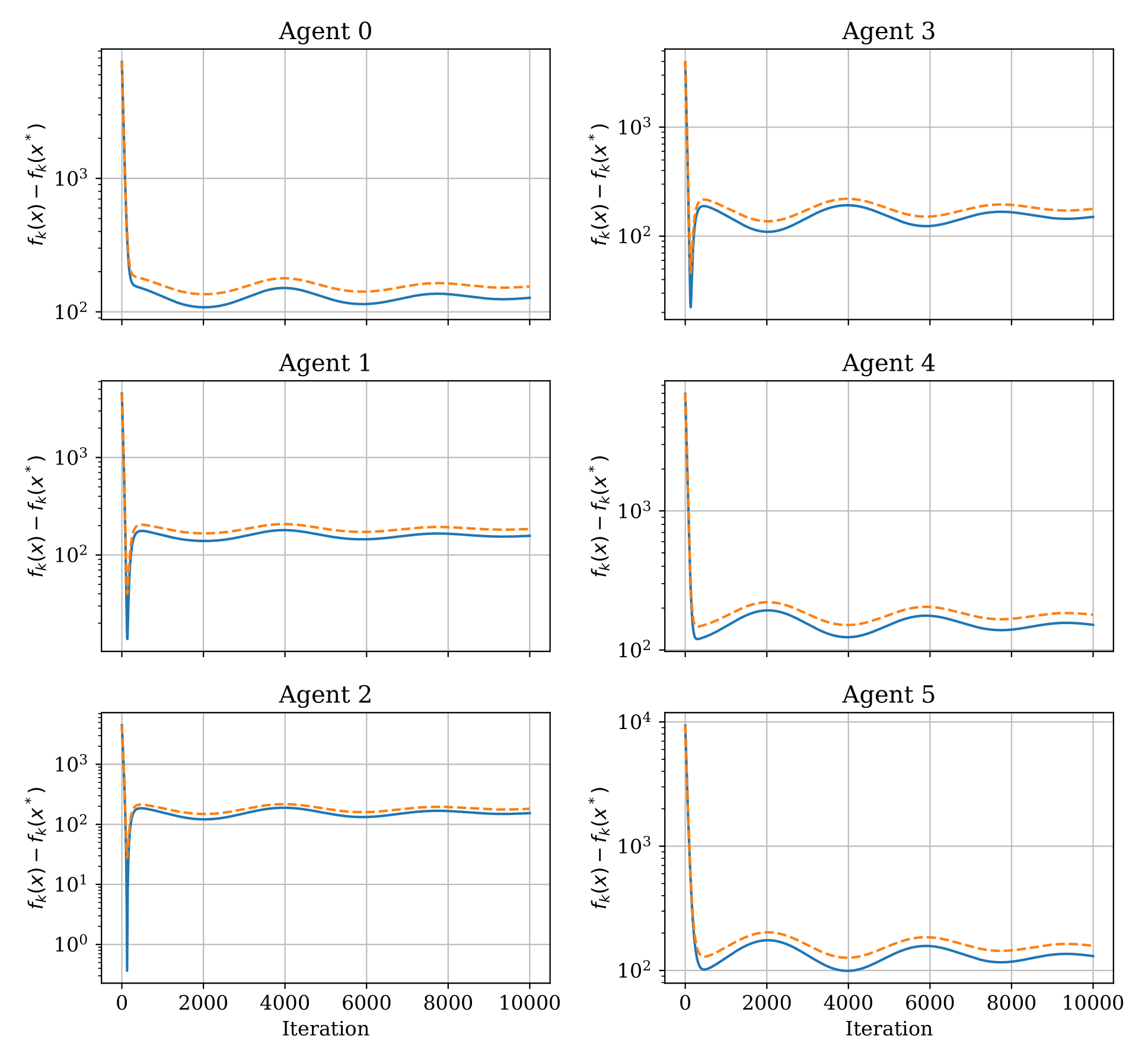}
 \caption{Minimization error with no formation control. \label{fig:minErrorNoForm}}
\end{figure}

\section{Conclusion}

In this paper we consider a formation of agents tracking the optimum of a time varying function $f_k$ with no gradient information. At each iteration, the agents take measurements, compute an approximate descent direction, and converge to a neighbourhood of the optimum. We derive bounds on the neighbourhood of convergence, as a function of the error in the gradient estimate, using minimal assumptions on the time-varying function. As the gradient approximation is constructed in a decentralised way, formation control is used to encourage the agents to formations which improve the gradient estimates, while not overwhelming the task of minimizing the source function $f_k$. We show that the formation control remains within a bounded distance of the optimal formation, and the implications to the convergence of the network to the minima of $f_k$. In the future, a more flexible formation control approach with convergence guarantees as well as hardware experiments will be investigated.

\bibliographystyle{IEEEtran}
\bibliography{IEEEabrv,cdc}
\appendix 


\subsection{Proof of Theorem~\ref{thm:formationThm}\label{app:proofFormationThm}}
We first show that the trajectory of the global potential function $\phi(x)$ is a sum of local information for each agent $i$. Then we prove that if the potential function has violated the upper bound, i.e. $\phi^i(x_k)\geq \Phi^i$, then $\phi^{i}(x_{k+1}) - \phi^{i}(x_k) \leq 0$ along trajectories. Then we show that if $\phi^i(x_k)\leq \Phi^i$, then $\phi^{i}(x_{k+1}) - \phi^{i}(x_k)$ is bounded, which finally gives that $\phi(x_k)$ is bounded for all $k$. Let $A_k$ be a diagonal matrix such that $A^i = \alpha^i_k$.

Begin with the definition of a Lipschitz continuous gradient for the global potential function $\phi(x)$
\begin{align*}
\phi(x_{k+1}&) - \phi(x_k) \leq \nabla \phi(x_k)^T (x_{k+1}-x_k)\\
&\qquad \qquad \qquad \qquad \qquad + \frac{L_\phi}{2} ||x_{k+1}-x_k||^2,\\
&= \nabla \phi(x_k)^T (-A_k p_k) + \frac{L_\phi}{2} ||A_k p_k||^2,\\
&= \sum_{i\in{\mathcal{V}}} \nabla_{i}\phi(x_k)^T(-\alpha^i_k p^{i}_k) + \frac{L_\phi}{2}||\alpha^i_k p^{i}_k||^2,
\end{align*}
which is equivalent to the sum of the Lipschitz conditions for each of the agents $i$ individually, although using the Lipschitz constant $L_\phi$ of the global function. Given Assumption~\ref{ass:phiProperties}, each agent has all the information required to compute their local Lipschitz bound. We proceed with the analysis of the Lipschitz bound of a single agent,
\begin{align}
\phi^{i}(x_{k+1}) - \phi^{i}(x_k) &\leq \nabla_{i} \phi(x_k)^T(-\alpha^i_k p^i_k) + \frac{L_\phi(\alpha^i_k)^2}{2}||p^i_k||^2 \nonumber \\
\intertext{By the choice of $\alpha^i_k\leq \frac{1}{L_\phi}$ we have $L_\phi(\alpha^i_k)^2\leq \alpha^i_k$,}
\phi^{i}(x_{k+1}) - \phi^{i}(x_k) &\leq -\alpha^i_k \nabla_{i} \phi(x_k)^Tp^i_k + \frac{\alpha^i_k}{2}||p^i_k||^2\nonumber \\
&\leq \alpha^i_k (\frac{p^i_k}{2} - \nabla_{i}\phi(x_k))^Tp^i_k\nonumber \\
\intertext{Expanding $p^i_k=\lambda \Lambda^{i}(x_k) + (1-\lambda) \nabla_{i} \phi(x_k)$ and simplifying}
\phi^{i}(x_{k+1}) - \phi^{i}(x_k) &\leq \frac{\alpha^i_k}{2}\lambda^2||\Lambda^{i}(x_k)-\nabla_{i}\phi(x_k)||^2\nonumber\\
&\qquad \qquad \qquad - \frac{\alpha^i_k}{2}||\nabla_{i} \phi(x_k)||^2. \label{eq:individualCondition}
\end{align}

Now suppose that $\phi^{i}(x_k) \geq \Phi^{i}$. We directly have
\begin{align*}
\lambda &\leq \frac{||\nabla_{i} \phi(x_k)||}{||\Lambda^{i}(x_k) - \nabla_{i} \phi(x_k)||} \frac{\sigma(\Phi^{i})}{\sigma(\phi^{i}(x_k))} \\
\lambda &\leq \frac{||\nabla_{i} \phi(x_k)||}{||\Lambda^{i}(x_k) - \nabla_{i} \phi(x_k)||}
\end{align*}
Substituting this bound into~\eqref{eq:individualCondition}, we obtain
\begin{align*}
\phi^{i}(x_{k+1}) - \phi^{i}(x_k) &\leq \frac{\alpha^i_k}{2}\lambda^2||\Lambda^{i}(x_k)-\nabla_{i}\phi(x_k)||^2\\
&\qquad \qquad \qquad \quad - \frac{\alpha^i_k}{2}||\nabla_{i} \phi(x_k)||^2 \\
&\leq 0.
\end{align*}
This shows that if $\phi^{i}(x_k) \geq \Phi^{i}$, then $\phi^{i}(x_{k+1}) - \phi^{i}(x_k) \leq 0$ along trajectories. If we assume that $\phi^{i}(x_k) \leq \Phi^{i}$ we have
\begin{align*}
\phi^{i}(x_{k+1}) - \phi^{i}(x_k) &\leq \frac{\alpha^i_k}{2}\lambda^2||\Lambda^{i}(x_k)-\nabla_{i}\phi(x_k)||^2\\
&\qquad \qquad \qquad \qquad - \frac{\alpha^i_k}{2}||\nabla_{i} \phi(x_k)||^2 \\
&\leq \frac{\alpha^i_k}{2}||\Lambda^{i}(x_k)-\nabla_{i}\phi(x_k)||^2\\
&\qquad \qquad \qquad \qquad - \frac{\alpha^i_k}{2}||\nabla_{i} \phi(x_k)||^2 \\
&\leq c.
\end{align*}
Then the potential value of agent $i$ will remain bounded $\phi^i(x_k) \leq \Phi^i + c$, the global potential function $\phi(x_k)$ will remain bounded by the sum over all agents. 

\subsection{Proof of Theorem~\ref{prop:neighbourhood}\label{app:proofProp}}

The agent identifying superscript is suppressed in this proof, as all calculations correspond to a single agent $i$. Recall, as stated in Assumption~\ref{ass:drift}, that there exist scalars $\eta_0$ and $\eta^*$ which bound functions drift over time. By~\eqref{eq:dynWgrad} and Assumption~\ref{ass:Lipschitz},
\begin{align*}
&f_{k}(x_{k+1}) - f_k(x_k) \leq \nabla f_k(x_k)^T (x_{k+1} - x_k)\\
&\qquad\qquad\qquad\qquad\qquad\qquad + \frac{L}{2}||x_{k+1} - x_k||^2,\\
\begin{split}
&= -\alpha \nabla f_k(x_k)^T \nabla f_k(x) - \alpha \nabla f_k(x_k)^T\bar{\delta} \\
&\qquad\qquad\qquad\qquad\qquad\qquad + \frac{(\alpha)^2 L}{2}||\nabla f_k(x) + \bar{\delta}||^2,
\end{split}\\
\begin{split}
&= -\frac{\alpha}{2} ||\nabla f(x)||^2 +\alpha \left (\frac{\alpha L}{2}-\frac{1}{2} \right )||\nabla f_k(x_k) + \bar{\delta}||^2 \\
&\qquad\qquad\qquad\qquad\qquad\qquad + \frac{\alpha}{2}||\bar{\delta}||^2.
\end{split}
\end{align*}
Restricting $\alpha$ to lie in the interval $\alpha\in[0,\frac{1}{L}]$ such that $\frac{\alpha L}{2}-\frac{1}{2}\leq 0$ and using Assumption~\ref{ass:polyak}, we have
\begin{align*}
f_{k}(x_{k+1}) - f_k(x_k) &\leq -\frac{\alpha}{2} ||\nabla f(x)||^2 + \frac{\alpha}{2}||\bar{\delta}||^2, \\ 
&\leq -s\alpha(f_k(x_k)-f_k^*) + \frac{\alpha}{2}\bar{\delta}^2,
\end{align*}
and therefore we have
\begin{align*}
0 &\leq -s\alpha(f_k(x_k)-f_k^*) + f_k(x_k) - f_{k}(x_{k+1}) + \frac{\alpha}{2}\bar{\delta}^2
\end{align*}
Adding $f_{k+1}(x_{k+1})-f_{k+1}^*$ to both sides, and using the scalar bounds from Assumption~\ref{ass:drift}, we obtain
\begin{align}
f_{k+1}(x_{k+1})-f_{k+1}^* &\leq (1-s\alpha)(f_k(x_k)-f_k^*) + \eta_0 \nonumber\\
& \qquad \qquad +\eta^* + \frac{\alpha}{2}\bar{\delta}^2, \label{eq:recursive}
\end{align}

The recursive relation defined in~\eqref{eq:recursive} can be expressed analytically as
\begin{align}
  &f_{k+1}(x_{k+1})-f_{k+1}^* \leq \frac{\eta_0 + \eta^*}{s \alpha} + \frac{\bar{\delta}^2}{2s} \nonumber\\ 
  &\qquad+(1-s\alpha)^{k+1}\left (f_0(x_0)-f_0^* - \frac{\eta_0 + \eta^*}{s \alpha} - \frac{\bar{\delta}^2}{2s} \right ). \label{eq:analytic}
\end{align}
Finally, given that Assumption~\ref{ass:polyak} is equivalent to $$f_k(x_k) - f_k^* \leq 2s ||x_k-\bar{x}_k||^2,$$ where $\bar{x}_k$ is the projection of $x_k$ onto $\mathcal{X}_k^*$, we obtain
\begin{align}
  &||x_k-\bar{x}_k||^2 \leq \frac{\eta_0 + \eta^*}{2s^2\alpha} + (1-(1-s\alpha)^{k+1})\frac{\bar{\delta}^2}{4s^2} \nonumber\\
  &\qquad+(1-s\alpha)^{k+1}\left (||x_k-\bar{x}_0||^2 - \frac{\eta_0 + \eta^*}{2s^2 \alpha} \right ).\label{eq:issCondition2}
\end{align}
Therefore, as $1-\alpha s\leq 1$, the agent will reach an $M = \frac{\eta_0 + \eta^*}{2\alpha s^2} + \frac{\bar{\delta}^2}{4s^2}$ neighbourhood of the optimiser $\mathcal{X}_k^*$ as the time steps $k\rightarrow\infty$.

\end{document}